\documentclass[11pt,reqno]{amsproc}
\usepackage[german,english]{babel}  
\usepackage[latin1]{inputenc} 
\usepackage{graphicx, color}
\usepackage{amssymb}
\numberwithin{equation}{section}
\newtheorem{theorem}{Theorem}[section]

\newtheorem{lemma}[theorem]{Lemma}

\def\tr{\mathop{\textnormal{tr}}}

\newcommand{\supp}{\text{\rm{supp}}}
\newcommand{\R}{\mathbb{R}}

\newcommand{\Z}{\mathbb{Z}}

\renewcommand{\P}{\mathbb{P}}
\newcommand{\E}{\mathbb{E}}

\newcommand{\Zd}{\mathbb{Z}^{d}}
\newcommand{\ZN}{\mathbb{Z}^{Nd}}

\renewcommand{\L}{\Lambda}

\newcommand{\norm}[1]{||\,#1\,||}
\newcommand{\abs}[1]{|\,#1\,|}

\def\<{\langle}
\def\>{\rangle}

\def\p0{\psi_0}                              

\def\1Ll{{1 \over {\vert \Lambda_L \vert}}}

\def\j1d{{(1+\vert j \vert )}^{-(d+2)}}
\def\x1a{{(1+ \vert x\vert )}^{-\alpha}}
\def\2j1a{{(1+ \vert j\vert )}^{-\alpha}}

\begin{document}

\title[A multi-particle Wegner estimate]
{A Wegner estimate\\[3mm] for multi-particle random Hamiltonians\\[5mm]}

\author[Werner Kirsch]{Werner Kirsch\vspace{0.5mm}\\  Institut für Mathematik\\ Ruhr-Universität Bochum}
\address{Institut f\"ur Mathematik and SFB TR 12,
Ruhr-Universit\"at Bochum, \goodbreak D-44780 Bochum, Germany}
\email{werner.kirsch@ruhr-uni-bochum.de}



\begin{abstract}
We  prove a Wegner estimate for a large class of multi-particle
Anderson Hamiltonians on the lattice. These estimates will allow us
to prove Anderson localization for such systems. A detailed proof of
localization will be given in a subsequent paper.
\end{abstract}

\maketitle

\section{Introduction\label{sec:intro}}

Wegner estimates originate in the famous paper \cite{Wegner}. There, Wegner proved among other things that the
integrated density of states for the Anderson Hamiltonian has a bounded density provided the probability distribution
of the random potential itself has a bounded density. This implies in particular an upper bound on the probability that
an Anderson Hamiltonian on a finite box has eigenvalues close to a given energy $E$.

Wegner's estimate play a key role in the Multiscale method to prove Anderson localization (see e.\;g. \cite{Froehlich} or
\cite{DK}). Only recently Bourgain and Kenig \cite{BourgKen} proved Anderson localization for a Bernoulli model without
an \emph{a priori} Wegner estimate; they prove a Wegner-type estimate inductively within the Multiscale scheme.

Wegner's original work was restricted to lattice models. However, the estimate was also proven for the continuum (see \cite{CHK}
for a recent rather optimal result and \cite{Veselic} for a review on this subject).

In this note we prove a Wegner estimate for a multi-particle Anderson model. In a subsequent paper we will also do multiscale analysis
for this model. The first Wegner estimate for a multi-particle random Hamiltonian was proved by Zenk \cite{Zenk}. Chulaevsky and
Suhov \cite{ChulaSuh} develop a multiscale analysis for certain (1-d) two body Hamiltonians. The Wegner estimate in this paper
requires strong conditions on the probability density of the random potential (e.\;g. analyticity). It was one of the motivations
of the present note to avoid these strong assumptions.

The method of proof applied here is close to Wegner's original idea and was developed from the paper \cite{Kirsch}. Note that
there is a refinement of this method by Stollmann \cite{Stollmann} which is likely to work in the multi-particle case as well.

We note that the method presented in this paper will also work for
alloy-type models in the continuous case. The necessary changes can
be read off from the paper \cite{Kirsch}. However, in the continuous
case we get the volume factor of the bound with an exponent 2. This
suffices to do a multiscale analysis, but it gives no result for the
regularity of the integrated density of states.

\section{Models and Results}

We will deal with a system of N interacting particles on a lattice $\Z^d$. We consider these particles on the full Hilbert space,
disregarding Fermionic or Bosonic symmetry. Physically speaking we deal with distinguishable particle. Since the full Hilbert
space is a direct sum of the irreducible subspaces with respect to $S_N$-symmetry (including the totally symmetric and the totally
antisymmetric subspaces) the Wegner estimates for Fermions and Bosons follow immediately from the result on the full space.

The one-particle Hilbert space we consider is $\ell^2(\mathbb{Z}^{d})$ and the
Hilbert space for $N$ (distinguishable) particles is consequently $\ell^2(\ZN)$. Any (bounded) operator $A$ on these Hilbert
spaces is uniquely defined through its matrix elements $A(x,y)\,=\,(\delta_x,A\,\delta_y)$ where $\delta_z$ is the vector
in $\ell^2$ with component $1$ at lattice site $z$ and $0$ otherwise.

We write the lattice site $x\in\ZN$ as $x=(x_1,\ldots,x_N)$, where $x_i\in\Zd$ denotes the the coordinates of the
$i^{th}$ particle.

Each particle (with coordinates $\xi$) is subject to a random potential $v_\omega(\xi)$ which is the same for all particles.
The random potential $v_\omega(\xi)$ consists of independent identically distributed random variables. Throughout we assume that
the distribution of the $v(\xi)$ has a bounded density $\rho(v)$. We denote the underlying probability measure by $\P$ and the
expectation with respect to $\P$ by $\E$.

The kinetic energy operator for one particle is given by:

\begin{equation}
h_0\,u(\xi)~=~\sum_{|n|=1,~n\in\Z^d}\;u(\xi+n)\qquad\qquad \xi\in\Z^d
\end{equation}

the single particle Hamiltonian is consequently:

\begin{equation}
h_\omega~=~h_0\;+\;v_\omega
\end{equation}

If $h$ is a one-particle operator acting in $\ell^2(\Z^d)$ we denote by $h^{(i)}$ the corresponding operator on $\ell^2(\Z^{Nd})$
acting on the $i^{th}$ particle only, more precisely: If $h$ has matrix elements $h(\xi,\eta)$ then:

\begin{equation}
h^{(i)}\,u(x_1,\ldots,x_n)~=~\sum_{\eta\in\Z^d}\;h(x_i,\eta)\,u(x_1,\ldots,x_{i-1},\eta,x_{i+1},\ldots,x_N)
\end{equation}

In other words,
\begin{equation}
h^{(i)}~=\underbrace{\mathbf{1}_{\ell^2(\Z^d)}\otimes \ldots \otimes\mathbf{1}_{\ell^2(\Z^d)}}_{i-1 \textnormal{ times}}
\otimes \;\,h\; \otimes \underbrace{\mathbf{1}_{\ell^2(\Z^d)}\otimes \ldots \otimes\mathbf{1}_{\ell^2(\Z^d)}}_{n-i-1 \textnormal{ times}}
\end{equation}

The N-particle Hamiltonian without interaction is defined by:

\begin{equation}
    H_{\omega,\,0}~=~\sum_{i=1}^N\;h^{(i)}
\end{equation}

The interaction term $U$ can be a rather general function on $\Z^{Nd}$. We assume it to be bounded for simplicity.
We also suppose that $U$ is a deterministic function, it would be sufficient for our purpose to have $U$ independent of $v_\omega$.
In most cases $U$ is a pair potential of the form $U(x)=\sum_{i\not=j}u(x_i-x_j)$.

The N-particle Hamiltonian with interaction $U$is then given by:

\begin{equation}
    H_{\omega,\,U}~=~H_{\omega,\,0}\;+\;U
\end{equation}

We will deal with this operator restricted to a bounded (hence finite) domain $\L$. The number of elements
of $\L$ will be denoted by $\abs{\L}$.

We call a subset $R$ of $\Z^d$ a rectangle if
\begin{equation}
R=\{\xi\in\Z^d\;|\;L_\nu\leq\xi_\nu\leq M_\nu~~\textnormal{for }\nu=1\ldots N\,\}
\end{equation}

A rectangular domain in $\Z^{Nd}$ is a set $\Lambda$ of the form:

\begin{equation}
    \Lambda~=~\Lambda_1\,\times\,\Lambda_2\times\ldots\times\Lambda_N
\end{equation}

where the $\L_i$ are rectangles in $\Z^d$. We use the notation $\Pi_i(\L)=\L_i$. We call a
rectangular domain $\L$ \emph{regular} if for all $i,j=1,\ldots N$ either $\L_i\cap \L_j=\emptyset$ or $\L_i= \L_j$.

For any subset $\L$ of $\ZN$ we define the operator $H^\L=H^\L_{\omega,\,U}$ by its matrix
elements:

\begin{equation}
    H^\L\,(x,y)~=~H_{\omega,\,U}(x,y)\hspace{1.5cm} \textnormal{for }x,y\in\L
\end{equation}

The main result of this note is the following Wegner estimate for multi-particle operators:

\begin{theorem}\label{th:Wegner} If $\L$ is a regular rectangular domain then
\begin{equation}\label{Wegner}
   \P\,\Big(\,\textnormal{dist}\left(\,\sigma\left(H^\L\right),\,E\,\right)\;<\kappa\;\Big)~~\leq~~C\;\norm{\rho}_\infty\;|\,\L\,|\;\kappa
\end{equation}
\end{theorem}

The assumption of regularity of the set $\L$ can be avoided. However, the proof is more transparent with this assumption. The proof
of Anderson localization by multiscale analysis which we will present in a forthcoming paper will deal with regular domains only.

The proof of Theorem \ref{th:Wegner} implies also that the
integrated density of states has a bounded density. This result can also read off from the
explicitly known form of the integrated density of states (see \cite{klozen}).

\section{Proof}
We prove Theorem \ref{th:Wegner}. Let $\L=\L_1\times\L_2\dots\times\L_N$. We may assume that:
\begin{align}\label{L1} &\L_1=\L_2=\dots=\L_K\\ \textnormal{and\qquad} &\L_1\cap\L_i=\emptyset \quad
\textnormal{for all } i>K\label{L2}
\end{align}

We denote the eigenvalues of $H^\L$ by $E_n=E_n(H^\L)$. We order them so that $E_1\leq E_2\leq\dots$ and repeat any eigenvalue
according to its multiplicity. The eigenvalue counting function is denoted by:
\begin{equation}
N(H^\L,E)=\#\{E_n(H^\L)\leq E\}
\end{equation}

We will need the following Lemma:

\begin{lemma}\label{lem:de} Suppose (\ref{L1}) and (\ref{L2}) hold. Denote by $v(\xi)$ the value of the
random potential $v_\omega$ evaluated at the lattice site $\xi\in\Zd$. Then:
\begin{equation}\label{eq:de}
   \sum_{\xi\in\L_1} \frac{\partial E_n(H^\L)}{\partial v(\xi)}~=~K
\end{equation}

\end{lemma}

\begin{proof}
Set $V(x)=\sum_{i=1}^N\,v(x_i)$ Then for $\xi\in\L_1$:
\begin{equation}
\frac{\partial V}{\partial v(\xi)}\,(x_1,\dots,x_N)~=~\sum_{i=1}^K\;\delta_{\xi\,x_i}
\end{equation}
Hence for each $(x_1,\dots,x_N)\in\L$ we have
\begin{equation}\label{eq:sdv}
\sum_{\xi\in\L_1}\;\frac{\partial V}{\partial v(\xi)}\,(x_1,\dots,x_N)~
=~\sum_{\xi\in\L_1}\;\sum_{i=1}^K\;\delta_{\xi\,x_i}~=~K
\end{equation}

Let us denote by $\psi_n$ the normalized eigenfunction of $H^\L$ for the eigenvalue $E_n=E_n(H^\L)$.
The Feynman-Hellman-Theorem tells us that:
\begin{align}
\frac{\partial E_n}{\partial v(\xi)}~=~\langle\,\psi_n,\,\frac{\partial V}{\partial v(\xi)}\,\psi_n\,\rangle~
=~\sum_{x\in\L}\;|\,\psi_n(x)\,|^2\;\frac{\partial V}{\partial v(\xi)}
\end{align}
Thus from (\ref{eq:sdv}) we obtain:
\begin{align}
\sum_{\xi\in\L_1}\,\frac{\partial E_n}{\partial v(\xi)}~
&=~\sum_{x\in\L}\;|\,\psi_n(x)\,|^2\;\big(\sum_{\xi\in\L_1}\,\frac{\partial V(x)}{\partial v(\xi)}\,\big)\\
&=~K\;\sum_{x\in\L}\,|\,\psi_n(x)\,|^2~=K
\end{align}
since $\psi_n$ is normalized.
\end{proof}
Let $\varphi$ be an increasing $C^\infty-$function on $\R$, $0\leq \varphi \leq 1$ with $\varphi=1$ on $(\kappa,\infty)$
and $\varphi=0$ on $(-\infty,-\kappa)$.

Then:
\begin{align}
&\P\,\Big(\,\textnormal{dist}\left(\,\sigma\left(H^\L\right),\,E\,\right)\;<\kappa\;\Big)\\
\leq~ &\E\,\Big(\,N(H^\L,E+\kappa)\,-\,N(H^\L,E-\kappa)\,\Big)\\
=~ &\E\,\Big(\,\tr\,\big(\chi_{(E-\kappa,E+\kappa]}(H^\L)\, \big)\,\Big)\\
\leq~ &\E\,\bigg(\,\tr\,\Big(\,\varphi(H^\L\,-E\,+2\kappa)\,- \varphi\,(H^\L\,-E\,-2\kappa\,\Big)\,\bigg)\\
\leq~ &\E\,\bigg(\,\int_{-2\kappa}^{2\kappa}\;\tr\,\Big(\,\varphi'(H^\L\,-E\,+t)\,\Big)\,dt\,\bigg)\\
\intertext{by Lemma \ref{lem:de}:}
\leq~ &\frac{1}{K}\;\sum_n\, \int_{-2\kappa}^{2\kappa}\;\E\,\bigg(\,
\varphi'\Big(E_n(H^\L)\,-E\,+t\Big)\;\sum_{\xi\in\L_1}\,\frac{\partial E_n(H^\L)}{\partial v(\xi)}\,\bigg)\,dt\\
\leq~ &\sum_n\, \int_{-2\kappa}^{2\kappa}\;\sum_{\xi\in\L_1}\;\E\,\bigg(\,
\frac{\partial\,\varphi\Big(E_n(H^\L)\,-E\,+t\Big)}{\partial v(\xi)}\; \bigg)\,dt\label{est:W}
\end{align}
Since $\E$ is a product measure we can split it into an integration over $v(\xi)$ which we write as $\int\,\cdot\;\rho(v)\,dv$
and the expectation with respect to the other random variables, which expectation we denote as $\E_{v(\xi)}^-$ .

With this notation (\ref{est:W}) equals:
\begin{align}
&\sum_n\,\int_{-2\kappa}^{2\kappa}\,dt\;\sum_{\xi\in\L_1}\;\E_{v(\xi)}^-\,\Big(\,\int
\frac{\partial\varphi(E_m(H^\L-E+t)}{\partial v(\xi)}\;\rho\big(v(\xi)\big)\;dv(\xi)\,\Big)\\
&\leq~\norm{\rho}\;\sum_n\,\int_{-2\kappa}^{2\kappa}\,dt\;\sum_{\xi\in\L_1}\;\E_{v(\xi)}^-\,\Big(\,\int
\frac{\partial\varphi(E_m(H^\L-E+t)}{\partial v(\xi)}\;dv(\xi)\,\Big)\label{est:fin}
\end{align}

By the fundamental theorem of calculus we have:

\begin{align}
&\int\,\frac{\partial\varphi\,\big(E_n(H^\L)-E+t\,\big)}{\partial v(\xi)}\;dv(\xi)\\
~=~~&\varphi\big(E_n(H^\L_{v(\xi)=\max})-E+t\big)\,
-\,\varphi\big(E_n(H^\L_{v(\xi)=\min})-E+t\big)
\end{align}

\noindent where $H^\L_{v(\xi)=\max}$ (resp. $H^\L_{v(\xi)=\min}$) denotes the operator $H^\L$ with the potential $v(\xi)$ set to its
maximal (resp. minimal) value, i.~e. with $v(\xi)=\sup\,(\supp(\rho))$ or $v(\xi)=\inf\,(\supp(\rho))$. Note that we include
the cases $\sup\,(\supp(\rho))=\infty$ and $\inf\,(\supp(\rho))=-\infty$.

Changing $v(\xi)$ from its minimal to its maximal value is a (positive) perturbation of rank at most $M=K\frac{|\L|}{|\L_1|}$. Thus:

\begin{equation}
E_n(H^\L_{v(\xi)=min})~\leq~E_n(H^\L_{v(\xi)=max})~\leq~E_{n+M}(H^\L_{v(\xi)=min})
\end{equation}

To estimate (\ref{est:fin}) we use the following simple Lemma:

\begin{lemma}\label{lem:simple}
Let $\varphi$ be a non decreasing function on $\R$ with $0\leq\varphi\leq 1$.
If $a_n$ and $b_n$ are non decreasing sequences satisfying $a_n\leq b_n\leq a_{n+M}$ for all $n$, then:
\begin{align}
\sum_n\;\big(\;\varphi(b_n)-\varphi(a_n)\;\big)~\leq~M
\end{align}
\end{lemma}

Combining the above estimates we get:

\begin{align}
&\P\,\Big(\,\textnormal{dist}\left(\,\sigma\left(H^\L\right),\,E\,\right)\;<\kappa\;\Big)\\
\leq~&\norm{\rho}\;\sum_n\,\int_{-2\kappa}^{2\kappa}\,dt\;\sum_{\xi\in\L_1}\;\E_{v(\xi)}^-\,\Big(\,\int
\frac{\partial\varphi(E_m(H^\L-E+t)}{\partial v(\xi)}\;dv(\xi)\,\Big)\\
\leq~&\norm{\rho}\;\int_{-2\kappa}^{2\kappa}\,dt\;\sum_{\xi\in\L_1}\;\E_{v(\xi)}^-\,\sum_n\,
\Big(\varphi(E_n(H^\L_{v(\xi)=max})-\varphi(E_n(H^\L_{v(\xi)=min})\,\Big)\notag\\
\leq~&\norm{\rho}\;4\,\kappa\;|\L|
\end{align}


\begin{thebibliography}{999}
\bibitem{BourgKen} Bourgain J.; Kenig C.: On localization in the continuous
Anderson-Bernoulli model in higher dimension.
Invent. Math. {\bf 161} (2005), no. 2, 389--426

\bibitem{ChulaSuh} Chulaevsky V.; Suhov Yu.: Anderson localisation for interacting
multi-particle quantum system on Z. - Preprint, Universite de Reims, (version of Sept. 2006)


\bibitem{CHK} Combes, J.M.; Hislop, P.D.; Klopp, F.:
\textit{An optimal Wegner estimate and its application to the global
continuity of the integrated density of states for random
Schr\"odinger operators}. Preprint.

\bibitem{DK} H. von Dreifus, A. Klein: A new proof of localization in
  the Anderson tight binding model.  {\em Commun. Math. Phys.} {\bf
    124} (1989), 285--299.

\bibitem{Froehlich} Fr\"ohlich, J.; Spencer, T.:
\textit{Absence of diffusion in the Anderson tight binding model for large disorder or low energy}.
  Commun. Math. Phys. {\bf 88}, 151--184 (1983).

\bibitem{Kirsch} Kirsch W.: Wegner estimates and Anderson localization
  for alloy-type potentials. {\em Math. Z.} {\bf 221} (1996),
  507--512.

\bibitem{klozen} Klopp F., Zenk H.: The integrated density of states for an interacting multielectron homogeneous model,
Preprint

\bibitem{Stollmann} Stollmann, P.:
\textit{Wegner estimates and localization for continuum Anderson models with
some singular distributions}. Arch. Math. {\bf 75},  307--311 (2000).

\bibitem{Veselic} Veselic, I.:
\textit{Integrated density and Wegner estimates for random Schr\"odinger operators}.
Spectral Theory of Schr\"odinger Operators (Mexico, 2001). Edited by
Rafael del Río et al.
Providence, RI: American Mathematical Society (AMS). Contemp.
Math. {\bf 340}, 97--183 (2004).

\bibitem{Wegner} Wegner, F.:
\textit{Bounds on the density of states in disordered systems}.
Z. Phys. B {\bf 44}, 9--15 (1981).

\bibitem{Zenk} Zenk, H.:
\textit{An interacting multielectron Anderson model}, Preprint

\end{thebibliography}
\end{document}